\def\year{2019}\relax
\documentclass[letterpaper]{article} 
\usepackage{aaai19}  
\usepackage{times}  
\usepackage{helvet} 
\usepackage{courier}  
\usepackage[hyphens]{url}  
\usepackage{graphicx} 
\usepackage{graphicx}  
\title{Causal Mediation Analysis with Multiple Treatments and Latent Confounders }

\author{Wei Li\textsuperscript{\rm 1, 2}, Chunchen Liu\textsuperscript{\rm 2}, Zhi Geng\textsuperscript{\rm 1}, John Murray\textsuperscript{3}\\
\textsuperscript{\rm 1}School of Mathematical Sciences, Peking University, Beijing, China\\
\textsuperscript{\rm 2} Department of Data Mining, NEC Laboratory, China\\
\textsuperscript{\rm 3} Department of Computer Science, San Jose State University, USA
}

\usepackage{amsmath}
\usepackage{graphicx,amssymb,booktabs,multirow}
\usepackage{anysize}
\usepackage{bm}
\usepackage{natbib,amsmath}
\usepackage{graphicx}
\usepackage{fancyhdr}
\usepackage{rotating}

\makeatletter
\newcommand*{\indep}{%
\mathbin{%
\mathpalette{\@indep}{}%
}%
}
\newcommand*{\nindep}{%
\mathbin{
\mathpalette{\@indep}{\not}
}%
}
\newcommand*{\@indep}[2]{%
\sbox0{$#1\perp\m@th$}
\sbox2{$#1=$}
\sbox4{$#1\vcenter{}$}
\rlap{\copy0}
\dimen@=\dimexpr\ht2-\ht4-.2pt\relax
\kern\dimen@
{#2}%
\kern\dimen@
\copy0 
}
\makeatother
\usepackage{color}
\usepackage{amsthm}

\makeatother
\newtheorem{theorem}{Theorem}

\newtheorem{condition}{Condition}
\usepackage{bbm}
\newcommand{\argmin}{\arg\!\min}

\newcommand{\E}{\textnormal{E}}

\usepackage{caption}
\usepackage{subcaption}

\def\E{\textnormal{E}}

\usepackage[all]{xy}
\usepackage{bm}
\usepackage{pifont}
\usepackage{stmaryrd}
\usepackage{pgf,tikz}
 \usepackage{tikz}
 \usetikzlibrary{shapes,arrows,positioning,calc}

\newenvironment{sequation*}{\begin{equation*}\small}{\end{equation*}}

\newenvironment{tequation*}{\begin{equation*}\tiny}{\end{equation*}}

\usepackage{mathrsfs}

\newtheorem{innercustomgeneric}{\customgenericname}
\providecommand{\customgenericname}{}
\newcommand{\newcustomtheorem}[2]{%
  \newenvironment{#1}[1]
  {%
   \renewcommand\customgenericname{#2}%
   \renewcommand\theinnercustomgeneric{##1}%
   \innercustomgeneric
  }
  {\endinnercustomgeneric}
}

\newcustomtheorem{customthm}{Theorem}
\newcustomtheorem{customlemma}{Lemma}
\newcustomtheorem{customassumption}{Assumption}
\tikzstyle{sum} = [draw, circle, minimum size=0pt, inner sep = 1.7pt]

\makeatletter
\newcommand{\mathleft}{\@fleqntrue\@mathmargin0pt}
\newcommand{\mathcenter}{\@fleqnfalse}
\makeatother

\usepackage{booktabs,caption,fixltx2e}
\usepackage[flushleft]{threeparttable}
\allowdisplaybreaks

\def\bZ{\mathbf{Z}}
\def\bz{\mathbf{z}}
\def\E{\textnormal{E}}

\begin{document}

\maketitle

\begin{abstract}
Causal mediation analysis is used to evaluate direct and indirect causal effects of a treatment on an outcome of interest
  through an intermediate variable or a mediator.
  It is difficult to identify the direct and indirect causal effects because the mediator cannot be randomly assigned in many real applications.
  In this article, we consider a causal model including latent confounders between the mediator and the outcome. We present sufficient conditions for identifying the direct and indirect effects and propose an approach for estimating them. The performance of the proposed approach is evaluated by simulation studies. Finally, we apply the approach to a data set of the customer loyalty survey by a telecom company.
\end{abstract}

\section{Introduction}

Randomized experiments are typically seen as a gold standard for evaluating the causal effect of a treatment on an outcome. Although the estimation of causal effect allows researchers to examine whether the treatment causally affects the outcome, it provides only a black-box view of causality and cannot tell us how and why such an effect occurs.
Mediation analysis  seeks to open up the black box and helps us to understand how the treatment impacts the outcome. In particular,
mediation analysis is an important tool
for evaluating direct and indirect causal effects of the treatment on the outcome through an intermediate variable or a mediator.
A traditional approach to mediation analysis, which was commonly used in social psychological research \citep{baron1986moderator,mackinnon2007mediation}, involves three regression models:
a regression equation of mediator on treatment, a regression equation of outcome on both treatment and mediator, and a regression model of treatment on outcome. This regression-based traditional approach  is often known as the linear structural equation modeling (LSEM) method.
Usually, the LSEM framework does not consider latent confounders which affect both the mediator
and the outcome in the models.
However, in real applications, the mediator cannot be randomly assigned to individuals,
and there may exist latent variables confounding the mediator-outcome relationship. The presence of such latent confounders often induces the non-identifiability of
the direct and indirect causal effects of treatment on outcome. Apart from this, another drawback of the LSEM framework is that it cannot offer a general definition of these causal effects that are applicable beyond specific statistical models. Alternatively, a large number of scholars adopted the potential outcome framework to define the direct and indirect causal effects
 in causal mediation analysis \citep{robins1992identifiability,vanderweele2009marginal,imai2010identification,vanderweele2013mediation,li2017identifiability}.

Causal mediation analysis distinguishes between natural and controlled effects which are defined for different purposes  \citep{pearl2001direct}. For example, the natural direct effect captures the effect of the treatment when one intervenes to set the mediator to its naturally occurring level, while the controlled direct effect arises after intervening the mediator to a fixed level, which is particularly relevant
for policy making and requires that both the treatment and mediator can be directly manipulated. The natural direct and indirect effects are more useful for understanding
the underlying mechanism by which the treatment operates. This is because
the total causal effect can be decomposed into the sum of these two natural effects.

The identifiability of direct and indirect causal effects requires the sequential ignorability assumption \citep{imai2010general}
 or some other similar assumptions \citep{pearl2001direct,vanderweele2009conceptual}.
The sequential ignorability assumption means that the treatment is randomly assigned and the mediator is also randomly assigned conditional on the assigned treatment and the measured covariates. Under the sequential ignorability assumption, the parameters in the LSEM approach can also have causal interpretations. However, this assumption is too stringent and  may not hold even in randomized experiments. It is because that there may exist some latent confounders between the mediator and outcome variables.
For example, blood pressure as a mediator between treatment and heart disease
cannot be randomly assigned to patients,
and there may be latent confounders (e.g., diets, habits and genes) affecting both blood pressure and heart disease.
To address such latent confounding problems, one possible way is to perform
a sensitivity analysis to evaluate how sensitive the result is to the violation of the sequential ignorability assumption \citep{vanderweele2010bias,imai2010general,li2017identifiability}.
In order to obtain identifiability results for the case with latent confounders,
 \cite{ten2007causal} proposed a linear rank preserving model approach for assessment of causal mediation effects, but they made some no-interaction assumptions
which are untestable.
Following this line, \cite{zheng2015causal} extended this model to a more general setting but still required some other untestable assumptions.
Another way of dealing with latent confounders is to use baseline covariates interacted with the random treatment assignment as an instrumental variable.
Specifically, it assumes that there exists a baseline covariate which interacts with the treatment in predicting the mediator but is not predictive to the outcome \citep{dunn2007modelling,albert2008mediation,small2012mediation}.

In this article, we focus on the identification and estimation of natural direct and indirect causal effects in causal mediation analysis. We propose an approach
 to dealing with multiple and correlated treatments for causal mediation analysis.
 We give the formal definitions of causal mediation effects for each treatment while accounting for possible correlations with other treatments. Besides,
 our approach can also be applicable to the case with latent confounders between mediator and outcome variables.
 We allow for interactions between treatments and covariates in both mediator and outcome models, and we utilize the information from multiple treatments to identify direct and indirect causal effects and obtain consistent estimates of direct and indirect effects.


\section{Preliminaries}\label{sec:preliminaries}

\subsection{Observed random variables}
Let $\bZ_i$ denote the observed treatments assigned to individual $i$
which is a vector with $J \geq 2$ treatment variables, i.e., $\bZ_i=(Z_{i1}, Z_{i2},\ldots, Z_{iJ})^\top$. Let $Y_i$ denote the observed outcome for individual $i$ and $M_i$ denote some observed intermediate variable on the causal path from the treatments to the outcome.
To streamline notation of the random variables, we suppress subscript $i$ for individual below.

The components of $\bZ$ can be correlated with each other through an unobserved common cause of them.
For each $j=1,\ldots,J$, assume that $\Pr(Z_{j}=z_j)>0$ for any $z_j$th treatment level with $z_j\in\{1,\ldots,K_j\}$. Both $Y$ and $M$ are assumed to be continuous and there may exist a latent confounder $U$ confounding the relationship between these two variables. In general, we may also consider $M$ as a vector of mediator variables and our results in this article can be straightforwardly generalized from the setting of a single mediator to the setting of multiple mediators.

\subsection{Potential outcomes and assumptions}

To formally define causal effects in causal mediation analysis, we make use of the concept of potential outcomes. Potential outcomes present the values of a outcome variable for each individual under varying levels of a treatment variable. We can observe only one of these potential outcomes but can never observe all of them because it is impossible for us to unwind time and go back and manipulate the individual to other treatment levels.

We first make the stable unit treatment value assumption (SUTVA). This assumption requires that the value of the outcome should not be affected by the manner of manipulations providing the same value for the treatment variable, that is,
there is only one version of the potential outcomes and there is no interference between individuals \cite{rubin1980comment}. The SUTVA allows us to uniquely define the potential values for the mediator $M(\bz)$ and the potential outcome $Y(\bz)$ if an individual were to receive treatment $\bZ=\bz$. Let $Y(\bz,m)$ denote the potential outcome for an individual that would occur if the treatment $\bZ$ were set to level $\bz$, and if the mediator $M$ were manipulated to level $m$. In contrast, let $Y(\bz,M(\bz^*))$ denote the potential outcome for an individual, where we do not specify the actual level of $M$, but set it to what it would have been if treatment had been $\bZ=\bz^*$. To connect the observed random variables with corresponding potential outcomes, we also make the consistency assumption, namely that $M(\bz)=M$, $Y(\bz) = Y$ if $\bZ=\bz$, and $Y(\bz,m)=Y$ if $\bZ=\bz$, $M=m$. According to this assumption, we note that  the observed outcome $Y$ is just one realization of the potential outcome $Y(\bz,m)$ with observed treatment level $\bZ = \bz$ and mediator level $M=m$.

We also assume that the underlying causal model corresponds to a {\it directed acyclic graph} (DAG, \cite{pearl2000causality}). The DAG is a useful tool for visual representations of qualitative causal relationships between the variables of interest. We show the corresponding DAG of this context in Figure \ref{fig:dag}.  Note that there are no causal relationships between the treatment variables in $\bZ$, but they may be associated with each other through some unobserved variable.
We use $C$ to denote this unobserved common cause of the treatment variables in $\bZ$, and it is assumed to be independent of the other variables  conditional on $\bZ$. When the treatment variables in $\bZ$ are randomly assigned,
$C$ is an empty set, and this assumption automatically holds.
The symbol `{\scalebox{0.7}{$\vdots$}}' in the DAG represents other undisplayed treatment variables $Z_j$'s,
each of which has the directed edges $C \rightarrow Z_j$, $Z_j \rightarrow M$, and $Z_j \rightarrow Y$.
\begin{figure}[h]
\begin{center}
\begin{tikzpicture}
\node[sum, text centered] (c) {$C$};
\node[above right = 1.5 of c, text centered] (z1) {$Z_1$};
\node[right = 1.1 of c, text centered] (d) {$\vdots$};
\node[below right = 1.5 of c, text centered] (j) {$Z_J$};
\node[below right = 1.5 of z1, text centered] (m) {$M$};
\node[right=2.4 of m, text centered] (y) {$Y$};
\node[sum, above = of $(m)!0.2!(y)$, text centered] (u) {$U$};

\draw[->, line width= 0.5] (c) --  (z1);
\draw[->, line width= 0.5] (c) --  (j);
\draw[->, line width= 0.5] (z1) --  (m);
\draw[->, line width= 0.5] (j) --  (m);
\draw[->, line width= 0.5] (z1) to  [out=90,in=90, looseness=0.5]  (y);
\draw[->, line width= 0.5] (j) to  [out=270,in=270, looseness=0.5]  (y);
\draw [->, line width= 0.5] (m) -- (y);
\draw [->, line width= 0.5] (u) -- (m);
\draw [->, line width= 0.5] (u) -- (y);
\end{tikzpicture}
 \end{center}
\caption{ \label{fig:dag}
A DAG depicts the causal relationships for mediation analysis,
where the variables in circle are unobserved.
}
\end{figure}
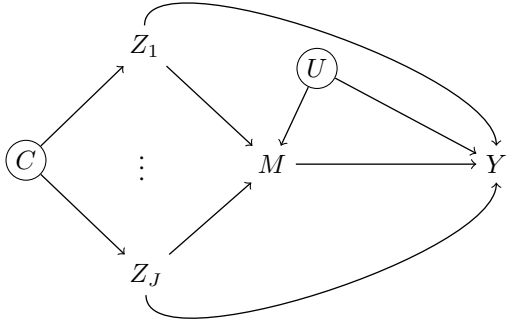

\subsection{Definitions of direct and indirect effects}

Using the nested potential outcomes notation, we can define the causal parameters of interest in this multiple-treatment model. We first describe the definition of the average causal effect in a single-treatment setting. We use the difference $\E\{Y(z)-Y(z^*)\}$ to represent the average causal effect of treatment level $z$ versus treatment level $z^*$.
We now extend this definition to the setting with multiple treatment variables. For notational simplicity, we let $\bZ_{-j}=(Z_1,\ldots,Z_{j-1},Z_{j+1},\ldots,Z_J)$ and $\bz_{-j}=(z_1,\ldots,z_{j-1},z_{j+1},\ldots,z_J)$. Then we
let $CTE(z_j,z_j^*\mid \bz_{-j})$ denote the conditional  total causal effect of $Z_j$ on $Y$ under two differing levels $z_j$ and $z_j^*$ of $Z_j$ while conditioning on the interventions $\bz_{-j}$ for the other treatments. The formulation is given as follows:
\begin{equation*}
  \begin{aligned}
    CTE(z_j,z_j^*\mid \bz_{-j})
    =\E\{Y(z_{j},\bz_{-j})\}
   -\E\{Y(z_j^*, \bz_{-j})\}.
    \vspace{-1em}
  \end{aligned}
\end{equation*}

In addition, we can also define  the average
total causal effect, $TE(z_j,z_j^*)$, which is free of other
treatment variables, by taking expectation of $CTE(z_j,z_j^*\mid \bZ_{-j})$ with respect to $\bZ_{-j}$, i.e.,
\begin{equation*}
  \begin{aligned}
    TE(z_j,z_j^*)=\E\{CTE(z_j,z_j^*\mid \bZ_{-j})\}.
  \end{aligned}
\end{equation*}

We next define the  conditional  natural direct and indirect effects of $Z_j$ on $Y$. Let $CNDE(z_j,z_j^*\mid \bz_{-j})$ denote the conditional  natural direct effect of  $Z_j$
under two different levels $z_j$ and $z_j^*$ while setting the other treatments to $\bz_{-j}$ and setting $M$ to the values attained under fixed treatment levels $\bz$.
We use $CNIE(z_j,z_j^*\mid \bz_{-j})$ to denote the conditional natural indirect effect which is defined as the difference between two averaged potential outcomes with treatments set to $z_j^*,\bz_{-j}$ and the mediator $M$ set to the values attained under differing treatment levels $z_j$ and $z_j^*$ for $Z_j$ and $\bz_{-j}$ for others.
Below we formally give their definitions:
\begin{equation*}
  \begin{aligned}
   CNDE(z_j,z_j^*\mid \bz_{-j})
     =& \E\{Y(z_j,\bz_{-j},M(\bz))\}
       \\&~-
    \E\{Y(z_j^*,\bz_{-j},M(\bz))\},\\
  CNIE(z_j,z_j^*\mid \bz_{-j})
     = &\E\{Y(z_j^*,\bz_{-j},M(z_j,\bz_{-j}))\}
      \\&~-
    \E\{Y(z_j^*,\bz_{-j},
    M(z_j^*,\bz_{-j}))\}.
  \end{aligned}
\end{equation*}
Similarly, we also give the the definitions of the average natural direct and indirect effects of $Z_j$ on $Y$ which do not depend on other treatment variables:
\begin{equation*}
  \begin{aligned}
    NDE(z_j,z_j^*)
    =~&\E\{CNDE(z_j,z_j^*\mid \bZ_{-j})\},\\
    NIE(z_j,z_j^*)
    =~&\E\{CNIE(z_j,z_j^*\mid \bZ_{-j})\}.
  \end{aligned}
\end{equation*}
Under the composition assumption \citep{pearl2000causality} that $Y(\bz)=Y(\bz,M(\bz))$, we immediately have the following decompositions:
\begin{equation*}
  \begin{aligned}
    CTE(z_j,z_j^*\mid \bz_{-j})
    =& CNDE(z_j,z_j^*\mid \bz_{-j})
     \\&~+CNIE(z_j,z_j^*\mid \bz_{-j}),\\
    TE(z_j,z_j^*)=&NDE(z_j,z_j^*)+NIE(z_j,z_j^*).
  \end{aligned}
\end{equation*}

\section{Methods}\label{sec:methods}
In this section, we first study the identification of the direct and indirect causal effects for some commonly-used models. We then provide an approach for estimation and also discuss inference procedures via  resampling techniques.

\subsection{Causal models}

We use potential outcomes notation to construct causal models that allow us to directly specify the causal effects of interest as functions of parameters in the models. We
consider the following causal models  including a latent confounder $U$ for potential outcomes, $M(\bz)$ and $Y(\bz,m)$, respectively:
\begin{eqnarray}
  \begin{aligned}
    M(\bz) &= g_M^c(\bz) + U + \epsilon(\bz),\\
    Y(\bz,m)&=g_Y^c(\bz) + \beta^c m + h(U) + \eta(\bz,m)\label{eqn:causalmodel},
  \end{aligned}
\end{eqnarray}
where $\E(U)=\E\{h(U)\}=0$, $g_M^c(\cdot)$, $g_Y^c(\cdot)$ and $h(\cdot)$ are unknown functions. Note that the causal models in~\eqref{eqn:causalmodel} allow the individual natural direct, indirect and total causal effects vary individual by individual. Similar models were also proposed in  \cite{lindquist2012functional}.
Here we use the superscript `$c$' to indicate causal parameters in models~\eqref{eqn:causalmodel} for potential outcomes to distinguish the parameters in SEMs for  observed variables.

We assume that $\E\{\epsilon(\bz)\}=0$ and $\E\{\eta(\bz,m)\}=0$
and that $\epsilon(\bz)$ and $\eta(\bz,m)$ are mutually independent and are also independent of $(\bZ,M,U)$ for all values of $\bz$ and the pair $(\bz,m)$. Then we can obtain $\E\{\eta(\bz,M(\bz^*))\}=0$ for all values of $\bz$ and $\bz^*$. With these conditions,
 we can directly write the average natural direct, indirect and total causal effects of $Z_j$ on $Y$ respectively as
\begin{equation}\label{eqn:causaleffectsofinterest}
  \begin{aligned}
   NDE(z_j,z_j^*)&=\E\{g_Y^c(z_j,\bZ_{-j})\}
     - \E\{g_Y^c(z_j^*,\bZ_{-j})\},\\
    NIE(z_j,z_j^*)&=\beta^c\big[\E\{g_M^c(z_j,\bZ_{-j})\}
     - \E\{g_M^c(z_j^*,\bZ_{-j})\}\big],\\
    TE(z_j,z_j^*)&=NDE(z_j,z_j^*)+NIE(z_j,z_j^*).
    \end{aligned}
\end{equation}
The identifiability of these causal effects relies on the identifiability of $g_M^c(\cdot)$, $g_Y^c(\cdot)$ and $\beta^c$. These parameters encode the causal relationships in models \eqref{eqn:causalmodel} for the potential outcomes.
Since the variables in the models are not observed,
these parameters cannot be estimated by the ordinary approaches for the models with observed variables. Hence, additional assumptions or (and) more feasible models are required to  identify and estimate these parameters.

\subsection{SEMs for observed variables}

According to the path diagram in Figure \ref{fig:dag}, we consider the SEMs for observed variables and a latent confounder $U$ as follows:
\begin{equation}
\begin{aligned}
   M(\bZ) &= g_M^s(\bZ) + U + \epsilon,\\
    Y(\bZ,M)&=g_Y^s(\bZ) + \beta^s M + h(U) + \eta\label{eqn:SEM},
    \end{aligned}
\end{equation}
where $\E(\epsilon\mid \bZ,U)=0$, $\E(\eta\mid \bZ,M,U)=0$,
$g_M^s(\cdot)$ and $g_Y^s(\cdot)$ are unknown functions, and the variables $M(\bZ)$ and $Y(\bZ, M)$ are observed because of the consistency assumption: $M(\bZ)=M$ and $Y(\bZ,M)=Y$.
The SEMs are built upon observed variables,
which are different from those for the potential outcomes defined in the previous subsection.
The superscript `$s$' denotes the parameters occurring in the SEMs,
and these parameters are in principle estimable from observed data.

The SEMs in~\eqref{eqn:SEM} assumes constant individual natural direct, indirect and total causal effects, which is more restrictive than that imposed in the causal models~\eqref{eqn:causalmodel}. However, we are interested in the average versions of these causal effects, and the average causal effects can be expressed as causal parameters $\{g_M^c(\cdot)$, $g_Y^c(\cdot)$, $\beta^c\}$ in models~\eqref{eqn:causalmodel} according to~\eqref{eqn:causaleffectsofinterest}. In general, the parameters $\{g_M^s(\cdot),g_Y^s(\cdot),\beta^s\}$ in the SEMs~\eqref{eqn:SEM}  are not equal to their counterparts
$\{g_M^c(\cdot)$, $g_Y^c(\cdot)$, $\beta^c\}$ in the causal models. However,
under the assumptions encoded by the DAG in Figure \ref{fig:dag},
we can establish the equality between these two sets of parameters.

\subsection{Equivalence between parameters in causal models and SEMs}
Under the models in~\eqref{eqn:SEM}, we can write the parameters $\{g_M^s(\cdot),g_Y^s(\cdot),\beta^s\}$ as
\begin{equation*}
  \begin{aligned}
    g_M^s(\bz) = \E\{M(\bZ)\mid \bZ =\bz\},~~~~~~~~~~~~~~~~~~~~~\qquad\qquad\qquad&\\
    \beta^s=\frac{1}{m-m^*}\big[\E\{Y(\bZ,M(\bZ))\mid M(\bZ)=m,\bZ=\bz,U\}&\\
    -\E\{Y(\bZ,M(\bZ))\mid M(\bZ)=m^*,\bZ=\bz,U\}\big],\qquad&\\
    g_Y^s(\bz) =\E\{Y(\bZ,M(\bZ))\mid M(\bZ)=m,\bZ=\bz, U\}\qquad\quad&\\
    -\beta^s m-h(U).\qquad\qquad\qquad\qquad\qquad\qquad\quad&
  \end{aligned}
\end{equation*}

Noting first from the DAG that $M(\bz)\indep \bZ$ and under the assumption that $\E(U+\epsilon\mid \bZ=\bz)=0$
for any value of $\bz$, we immediately have the following result:
\begin{equation*}
  \begin{aligned}
    g_M^s(\bz) = \E\{M(\bZ)\mid \bZ =\bz\}=\E\{M(\bz)\}=g_M^c(\bz).
  \end{aligned}
\end{equation*}
We also note that  the potential outcomes for the mediator and outcome are conditionally independent given the latent confounder $U$, i.e., $Y(\bz,m)\indep M(\bz)\mid U$ for any value of the pair $(\bz,m)$. Combining this with the ignorable treatment assignment assumption, we have
\begin{equation*}
  \begin{aligned}
    &\E\{Y(\bZ,M(\bZ))\mid M(\bZ)=m,\bZ=\bz,U\}\\
    =~&\E\{Y(\bz,m)\mid M(\bz)=m,U\}=\E\{Y(\bz,m)\mid U\}.
  \end{aligned}
\end{equation*}
Consequently,
\begin{equation*}
  \begin{aligned}
    \beta^s&=\frac{1}{m-m^*}\big[\E\{Y(\bz,m)\mid M(\bz)=m,U\}\\&~~~~~~~
     -\E\{Y(\bz,m^*)\mid M(\bz)=m^*,U\}\big]\\
    &=\frac{1}{m-m^*}\big[\E\{Y(\bz,m)\mid U\}-\E\{Y(\bz,m^*)\mid U\}\big]
    \\&=\beta^c.
  \end{aligned}
  \end{equation*}
In addition, it should also be noted that based on the previous results, we can also show the equality of parameter $g_Y^s(\cdot)$ with parameter $g_Y^c(\cdot)$ as follows:
\begin{equation*}
  \begin{aligned}
    g_Y^s(\bz) &=\E\{Y(\bz,m)\mid U\}-\beta^s m-h(U)=g_Y^c(\bz).
  \end{aligned}
\end{equation*}
Until now, we have shown the equivalence between parameters of the SEMs in \eqref{eqn:SEM} with the corresponding parameters of the causal models in \eqref{eqn:causalmodel}.
Since the SEMs include the unobserved variable $U$,
the parameters are still unidentifiable without additional conditions.
\subsection{Identification of parameters}

In this subsection, we give conditions under which the parameters of the SEMs are identifiable from observed data. Apparently, $g_M^s(\cdot)$ is identifiable due to the condition that $\E(U+\epsilon\mid \bZ)=0$ and can be written as follows:
\begin{equation*}
  g_M^s(\bz) = \E(M\mid \bZ=\bz).
\end{equation*}
In order to guarantee the identifiability of $\beta^s$ and $g_Y^s(\cdot)$, we impose the following condition.

\begin{condition}\label{cond:identification}
$\mathscr{G}_M^s:=\{g_M^s(\cdot)\}$ is a function space with finite dimension, and the space $\mathscr{G}_Y^s:=\{g_Y^s(\cdot)\}$ is a proper subspace of $\mathscr{G}_M^s$.
\end{condition}

The condition 1 implies that the number of the basis functions of $\mathscr{G}_M^s$ is finite and a proper subset of these basis functions generates the subspace $\mathscr{G}_Y^s$. This may be not a stringent condition and can be satisfied in a variety of cases. For example, suppose that $g_M^s(\bz)$ is a polynomial function of degree 2 in each component of $\bz$. Then, a linear function $g_Y^s(\bz)$ of the components satisfies Condition \ref{cond:identification}.

\begin{theorem}\label{thm:identification}
  Given the specified DAG in Figure~\ref{fig:dag} and models in \eqref{eqn:SEM},
  the parameters $\beta^s$ and $g_Y^s(\cdot)$ are identifiable under Condition \ref{cond:identification}.
\end{theorem}
\begin{proof}
  Suppose that the functions $\phi_1(\bz),\ldots,\phi_L(\bz)$ are composed of the basis functions of $\mathscr{G}_M^s$. Since $\mathscr{G}_Y^s$ is a proper subset of $\mathscr{G}_M^s$, we assume without loss of generality that the basis functions of $\mathscr{G}_Y^s$ are $\phi_1(\bz),\ldots,\phi_{L_1}(\bz)$, where $L_1<L$.

  By definition, there exist two sequences of real numbers $\{\alpha_l\}_{l=1}^L$ and $\{\gamma_l\}_{l=1}^{L_1}$ such that
  \begin{equation}\label{eqn:linearcombination}
  \begin{aligned}
    g_M^s(\bz)=\sum_{l=1}^L\alpha_l \phi_l(\bz),~~
    \text{and}~~
    g_Y^s(\bz)=\sum_{l=1}^{L_1}\gamma_l \phi_l(\bz).
    \end{aligned}
  \end{equation}
  As discussed earlier, $g_M^s(\bz)$ is identifiable, which  implies that the parameters $\{\alpha_l\}_{l=1}^L$ are also identifiable. Substituting the equation for $M$ into the equation for $Y$ in \eqref{eqn:SEM}, and replacing the expressions of $g_M^s(\bz)$ and $g_Y^s(\bz)$ with corresponding linear combinations as shown above, we obtain
  \begin{equation}\label{eqn:anothermodelforY}
    \begin{aligned}
      Y=&\sum_{l=1}^{L_1}(\gamma_l+\beta^s\alpha_l)\phi_l(\bZ)+
      \beta^s\sum_{L_1+1}^L\alpha_l\phi_l(\bZ)\\&~~~+\beta^sU
      +h(U)+\beta^s\varepsilon+\eta.
    \end{aligned}
  \end{equation}
  Because $\E\{\beta^sU+h(U)+\beta^s\varepsilon+\eta\mid \bZ\}=0$ and $\phi_1(\bz),\ldots,\phi_L(\bz)$ are linearly independent, we find that $\gamma_l+\beta^s\alpha_l$ ($l=1,\ldots,L_1$) and $\beta^s\alpha_l$ ($l=L_1+1,\ldots,L$) are identifiable. This, combined with the identifiability of $\{\alpha_l\}_{l=1}^L$, implies that $\beta^s$ and $\{\gamma_l\}_{l=1}^{L_1}$ are also identifiable. Thus, $g_Y^s(\bz)$ is identifiable.
\end{proof}

Based on the discussions about the equality between the parameters of the SEMs and the analogous parameters of the causal models in the previous subsection, we conclude from Theorem \ref{thm:identification} that
the parameters $g_M^c(\cdot)$, $g_Y^c(\cdot)$ and $\beta^c$ are also identifiable, which in turn implies the identifiability of the average natural direct, indirect and total causal effects of $Z_j$ on $Y$ according to \eqref{eqn:causaleffectsofinterest}.

\subsection{Estimation and inference}

In this subsection, we provide an approach for estimating the parameters of the SEMs in \eqref{eqn:SEM}. We also discuss the resampling-based procedures for inference.

Given the known basis functions $\phi_1(\bz),\ldots,\phi_L(\bz)$ of $\mathscr{G}_M^s$, we can express $g_M^s(\bz)$ and $g_Y^s(\bz)$ as linear combinations of them, which have been shown in \eqref{eqn:linearcombination}. Denote $\bm{\Phi}(\bz)=(\phi_1(\bz),\ldots,\phi_L(\bz))$, $\bm{\Phi}_1(\bz)=(\phi_1(\bz),\ldots,\phi_{L_1}(\bz))$,
$\bm{\alpha}=(\alpha_1,\ldots,\alpha_L)^\top$, $\bm{\gamma}=(\gamma_1,\ldots,\gamma_{L_1})^\top$. We then rewrite $g_M^s(\bz)$ and $g_Y^s(\bz)$ as
\begin{equation*}
  g_M^s(\bz) = \bm{\Phi}(\bz)\bm{\alpha}, ~~\text{and}~~ g_Y^s(\bz)=\bm{\Phi}_1(\bz)\bm{\gamma}.
\end{equation*}
By the least-square criterion,
the coefficient $\bm{\alpha}$ of the expansion of $g_M^s(\bz)$ can be determined by minimizing
\begin{equation*}
  \hat{\bm{\alpha}}=\argmin_{\bm{\alpha}} \sum_{i=1}^n \big\{M_i-\bm{\Phi}(\bZ_i)\bm{\alpha}\big\}^2.
\end{equation*}
Then the solution is given by
\begin{equation*}
  \hat{\bm{\alpha}}=\bigg\{\sum_{i=1}^n\bm{\Phi}(\bZ_i)^\top\bm{\Phi}(\bZ_i)\bigg\}^{-1}
  \bigg\{\sum_{i=1}^n\bm{\Phi}(\bZ_i)^\top M_i\bigg\},
\end{equation*}
which provides us an estimate of the coefficients of $g_M^s(\bz)$.

To derive estimates of $\beta^s$ and $\bm{\gamma}$, we utilize the idea of the generalized method of moments (GMM) \cite{hall2005generalized} and introduce more notations. Let $\bm{\alpha}_1=(\alpha_1,\ldots,\alpha_{L_1})^\top$, and the subvector being composed of the
remaining components in $\bm{\alpha}$ is denoted by $\bm{\alpha}_2$. Let $\bm{\Phi}_2(\bz)$ represent the subvector of $\bm{\Phi}(\bz)$ with components not contained in $\bm{\Phi}_1(\bz)$, i.e., $\bm{\Phi}(\bz)=\{\bm{\Phi}_1(\bz),\bm{\Phi}_2(\bz)\}$. Then, we
denote $\bm{Y}=(Y_1,\ldots,Y_n)^\top$, $\bm{\Phi}=\{\bm{\Phi}(\bZ_1)^\top,\ldots,\bm{\Phi}(\bZ_n)^\top\}^\top$, $\bm{\Phi}_1=\{\bm{\Phi}_1(\bZ_1)^\top,\ldots,\bm{\Phi}_1(\bZ_n)^\top\}^\top$, and $\bm{\Phi}_2=\{\bm{\Phi}_2(\bZ_1)^\top,\ldots,\bm{\Phi}_2(\bZ_n)^\top\}^\top$.
 For notational simplicity, we also let $\bm{\delta}=\bm{\gamma}+\beta^s\bm{\alpha_1}$. Using these notations, we now rewrite \eqref{eqn:anothermodelforY} as
\begin{equation*}
  Y = \bm{\Phi}_1(\bZ)\bm{\delta}+\beta^s\bm{\Phi}_2(\bZ)\bm{\alpha}_2+
  \beta^sU+h(U)+\beta^s\varepsilon+\eta.
\end{equation*}
In view of $\E\{\beta^sU+h(U)+\beta^s\varepsilon+\eta\mid \bZ\}=0$, it follows immediately that
\begin{equation}\label{eqn:residualexpectation}
  \E\big[\Phi(\bZ)^\top\{Y-\bm{\Phi}_1(\bZ)\bm{\delta}-\beta^s\bm{\Phi}_2(\bZ)\bm{\alpha}_2\}\big]
  =\bm{0}.
\end{equation}
We then estimate $(\bm{\delta}^\top,\beta^s)$ by minimizing the sum of the squares of sample analogues of \eqref{eqn:residualexpectation}
\begin{equation*}
\begin{aligned}
 &(\hat{\bm{\delta}}^\top,\hat{\beta^s})
  =\argmin_{(\bm{\delta}^\top,\beta^s)}(\bm{Y}-\bm{\Phi}_1
  \bm{\delta}-\beta^s\bm{\Phi}_2\bm{\alpha}_2)^\top\bm{\Phi}\bm{\Phi}^\top
 \\&~~~\qquad\qquad\qquad\qquad~~(\bm{Y}-\bm{\Phi}_1
  \bm{\delta}-\beta^s\bm{\Phi}_2\bm{\alpha}_2).
\end{aligned}
\end{equation*}
After solving the above optimization problem and substituting the estimator $\hat{\bm{\alpha}}_2$ into the solutions, we can easily obtain the estimates of $(\bm{\delta}^\top,\beta^s)$ with explicit forms. However, due to their complex expressions, we omit displaying them here. Let $\hat{\bm{\gamma}}=\hat{\bm{\delta}} - \hat{\beta^s}\hat{\bm{\alpha}}_1$, which gives us an estimate of $\bm{\gamma}$. In addition, all these estimators $\hat{\bm{\alpha}}$, $\hat{\beta^s}$ and $\hat{\bm{\gamma}}$ are consistent. Under some regularity conditions, they are also asymptotically normal. Substituting these estimators into \eqref{eqn:causaleffectsofinterest} and taking averages over samples, we obtain consistent estimators of the average causal effects of interest. Furthermore, by the Delta method, these consistent estimators are also asymptotically normal. Since analytical calculations of the asymptotic variances are difficult,
 we use a nonparametric bootstrap method to conduct inference.

\section{Simulation studies and application}\label{sec:experiments}

In this section, we first conduct simulation studies to evaluate the performance of the proposed estimators in finite samples. We then apply the proposed approach to a real customer loyalty data set.

\subsection{Simulation studies}

We consider two different simulation studies where data are generated by using only one and multiple (more than one) available instrumental variables, respectively. Each of the simulations is repeated 1000 times under different sample sizes $n=$ 500, 1000, and 2000. We
 report the results of estimated causal effects of interest by averaging over the 1000 replications.

In both of the simulation studies, we consider three treatment variables: $Z_1$, $Z_2$, and $Z_3$. Each of them is uniformly generated from $\{1, 2, 3\}$ with equal probability. The latent confounder $U$ is generated from $N(0, 1)$.
The mediator $M$ is generated from the following equation:
\begin{equation*}
  \begin{aligned}
    M =& Z_1+Z_2+Z_3+Z_1*Z_2 +Z_1*Z_3
    \\&~~+Z_2*Z_3
    +U+\epsilon_M,
  \end{aligned}
\end{equation*}
where $\epsilon_M\sim N(0,1)$. The only difference between the two simulation studies is the generation of the outcome $Y$. The details for the generation of $Y$ are described as follows:

\paragraph{ Simulation Study 1.} The outcome $Y$ is generated from
\begin{equation*}
  \begin{aligned}
    Y =& Z_1 + Z_2 + Z_3 + M + Z_1 * Z_2\\
    &~~+ Z_1 * Z_3
     + 2 * U + \epsilon_Y,
  \end{aligned}
\end{equation*}

where $\epsilon_Y\sim N(0,1)$, and is independent of $\epsilon_M$. Note that
$Z2*Z3$ is independent of $U$, associated with $M$, and has no direct effect  on $Y$ except through $M$. Thus, for this simulation, $Z2*Z3$ can be viewed as an instrumental variable.

We use both the proposed approach and  the traditional regression-based approach without considering the latent confounder $U$ for estimation.
To compare the results obtained from these two approaches,
we report the bias and standard error (SE) for estimators of  $NDE_j(2,1)$, $NIE_j(2,1)$ and $TE_j(2,1)$  based on 1000 replications for the sample sizes $n=500$, 1000, and 2000, respectively.
Here, the subscript `$j$' indicates the causal effects of the $j$th treatment variable and $j=1,2,3$.
Table~\ref{tab:sim1} displays the simulation results.


From Table~\ref{tab:sim1}, we can see that the estimates by our approach all have negligible biases for the small sample size 500. As sample size increases, both the biases and standard errors become much smaller.
In contrast, the estimates of the average natural direct and indirect effects obtained by the traditional approach
all have quite large biases even for large sample sizes.
It is because the traditional approach ignores the latent confounder $U$ which confounds the mediator-outcome
 relationship.
However, since  $U$ does not affect the relationship between treatments and outcome variables, the performance of the traditional approach for estimations on the average total effects behaves much better.
\begin{table}[h!]
\centering
\resizebox{0.48\textwidth}{!}{
\begin{threeparttable}
\def~{\hphantom{0}}
\caption{Results for the proposed approach and traditional approach with different sample sizes in {\bf Simulation Study 1}.} \label{tab:sim1}
\centering
\begin{tabular}{ccrrcrr}
  \hline
  & \multirow{ 2}{*}{$~n$} & \multicolumn{2}{c}{Proposed } & &\multicolumn{2}{c}{Traditional } \vspace{0.6mm}\\
  \cline{3-4}
  \cline{6-7}
  \addlinespace[0.9mm]
  & & Bias &SE  & &Bias &SE  \\
  \hline
  \addlinespace[1mm]
$Z_1$ &&&&&&\\
\cline{1-1}
  \multirow{ 3}{*}{$\textnormal{NDE}_1(2,1)$ = 5} &~500&0.057 &0.643 &  &-3.724 &0.280\\
  &1000 &0.010 & 0.464
   & &-3.719&0.189\\
  &2000 &0.008 & 0.326 & &-3.720 &0.133 \\ 
  \addlinespace[1mm]
  \multirow{ 3}{*}{$\textnormal{NIE}_1(2,1)$ = 5}&~500&-0.062 &0.649
  & &3.820 &0.297 \\
  &1000 &-0.014 &0.479
  & &3.815 &0.208  \\
  &2000 &-0.005 &0.335
  & &3.824 &0.149 \\ 
  \addlinespace[1mm]
  \multirow{ 3}{*}{$\textnormal{TE}_1(2,1)$ = 10}&~500 &-0.006 &0.210
  & &0.096 &0.212\\
  &1000 &-0.004 &0.146
  & &0.095 &0.148 \\
  &2000 &0.004 &0.101   & &0.104 &0.101\\
  \addlinespace[1mm]
  \hline
  \addlinespace[1mm]
  $Z_2$&&&&&&\\
  \cline{1-1}
\multirow{ 3}{*}{$\textnormal{NDE}_2(2,1)$ = 3}&~500 &0.067&0.632
& &-3.718 &0.284\\
  &1000 &0.012 &0.467
  & &-3.720 &0.188 \\
  &2000 &0.003 &0.324 &&-3.723 &0.130\\
  \addlinespace[1mm]
  \multirow{ 3}{*}{$\textnormal{NIE}_2(2,1)$ = 5} &~500 &-0.057 &0.653
  & &3.830&0.306 \\
  &1000 &-0.010 &0.479
  & &3.815 &0.218 \\
  &2000 &-0.008 &0.334
  &&3.819 &0.144\\
  \addlinespace[1mm]
  \multirow{ 3}{*}{$\textnormal{TE}_2(2,1)$ = 8}&~500 &0.011 &0.204
  & &0.111 &0.203\\
  &1000 &0.002 &0.138
  & &0.103 &0.138\\
  &2000 &-0.004 &0.096   & &0.097 &0.095\\
   \addlinespace[1mm]
  \hline
  \addlinespace[1mm]
  $Z_3$&&&&&&\\
  \cline{1-1}
  \multirow{ 3}{*}{$\textnormal{NDE}_3(2,1)$ = 3}&~500 &-0.052 &0.636
& &-3.728 &0.272\\
  &1000 &0.010 &0.466
  & &-3.718 &0.191 \\
  &2000 &0.009 &0.327 &&-3.717 &0.133\\
  \addlinespace[1mm]
  \multirow{ 3}{*}{$\textnormal{NIE}_3(2,1)$ = 5} &~500 &-0.064 &0.649
  & &3.818&0.300 \\
  &1000 &-0.014 &0.478
  & &3.815 &0.216 \\
  &2000 &-0.007 &0.334
  &&3.820 &0.149\\
  \addlinespace[1mm]
  \multirow{ 3}{*}{$\textnormal{TE}_3(2,1)$ = 8}&~500 &-0.012 &0.190
  & &0.090 &0.189\\
  &1000 &-0.004 &0.138
  & &0.097 &0.139 \\
  &2000 &0.002 &0.097  & &0.103 &0.096\\
  \hline
\end{tabular}
\end{threeparttable}}
\end{table}

\paragraph{ Simulation Study 2.} In this simulation study, we generate the outcome $Y$ from
\begin{equation*}
  \begin{aligned}
    Y = Z_1+Z_2+Z_3+M+2*U+\epsilon_Y,
  \end{aligned}
\end{equation*}
where $\epsilon_Y\sim N(0,1)$ and it is independent of $\epsilon_M$.
For this simulation, $Z_1*Z_2$, $Z_1*Z_3$ and $Z_2*Z_3$  can be
treated as instrumental variables.
We use both the proposed approach and the traditional approach for estimation, and the corresponding results are shown
  in Table~\ref{tab:sim2}.
The results are similar to those in  Table~\ref{tab:sim1},
 both of which support the consistency results of the proposed estimators and demonstrate the advantage of the proposed approach over the traditional one.

\begin{table}[h!]
\centering
\resizebox{0.48\textwidth}{!}{
\begin{threeparttable}
\def~{\hphantom{0}}
\caption{Results for the proposed approach and traditional approach with different sample sizes in {\bf Simulation Study 2}.} \label{tab:sim2}
\centering
\begin{tabular}{lcrrcrr}
  \hline
 & \multirow{ 2}{*}{$~n$} & \multicolumn{2}{c}{Proposed} & &\multicolumn{2}{c}{Traditional} \vspace{0.6mm}\\
  \cline{3-4}
  \cline{6-7}
  \addlinespace[0.9mm]
  & & Bias &SE  & &Bias &SE  \\
  \hline
\addlinespace[1mm]
  $Z_1$&&&&&&\\
  \cline{1-1}
  \multirow{ 3}{*}{$\textnormal{NDE}_1(2,1)$ = 1} &~500&0.001 &0.154 &  &-0.572 &0.134\\
  &1000 &0.000 & 0.107
   & &-0.569&0.096\\
  &2000 &-0.000 & 0.073 & &-0.570 &0.065 \\ 
  \addlinespace[1mm]
  \multirow{ 3}{*}{$\textnormal{NIE}_1(2,1)$ = 5}&~500&-0.002 &0.213
  & &0.918 &0.188 \\
  &1000 &-0.003 &0.151
  & &0.915 &0.138  \\
  &2000 &0.004 &0.106
  & &0.921 &0.095 \\ 
  \addlinespace[1mm]
  \multirow{ 3}{*}{$\textnormal{TE}_1(2,1)$ = 6}&~500 &-0.001 &0.213
  & &0.346 &0.188\\
  &1000 &-0.003 &0.134
  & &0.346 &0.129 \\
  &2000 &0.004 &0.093   & &0.351 &0.089\\
  \addlinespace[1mm]
  \hline
  \addlinespace[1mm]
  $Z_2$&&&&&&\\
  \cline{1-1}
\multirow{ 3}{*}{$\textnormal{NDE}_2(2,1)$ = 1}&~500 &0.006&0.150
& &-0.565 &0.133\\
  &1000 &0.002 &0.107
  & &-0.570 &0.095 \\
  &2000 &-0.004 &0.075 &&-0.574 &0.066\\
  \addlinespace[1mm]
  \multirow{ 3}{*}{$\textnormal{NIE}_2(2,1)$ = 5} &~500 &0.004 &0.216
  & &0.924&0.193 \\
  &1000 &0.001 &0.149
  & &0.920 &0.139 \\
  &2000 &0.001 &0.105
  &&0.918 &0.092\\
  \addlinespace[1mm]
  \multirow{ 3}{*}{$\textnormal{TE}_2(2,1)$ = 6}&~500 &0.010 &0.197
  & &0.359 &0.187\\
  &1000 &0.003 &0.134
  & &0.350 &0.129\\
  &2000 &-0.003 &0.092   & &0.344 &0.087\\
   \addlinespace[1mm]
  \hline
  \addlinespace[1mm]
  $Z_3$&&&&&&\\
  \cline{1-1}
  \multirow{ 3}{*}{$\textnormal{NDE}_3(2,1)$ = 1}&~500 &0.006 &0.150
& &-0.565 &0.133\\
  &1000 &0.002 &0.107
  & &-0.570 &0.095 \\
  &2000 &-0.004 &0.075 &&-0.574 &0.066\\
  \addlinespace[1mm]
  \multirow{ 3}{*}{$\textnormal{NIE}_3(2,1)$ = 5} &~500 &-0.004 &0.209
  & &0.916&0.187 \\
  &1000 &-0.003 &0.149
  & &0.915 &0.138 \\
  &2000 &0.001 &0.106
  &&0.919 &0.095\\
  \addlinespace[1mm]
  \multirow{ 3}{*}{$\textnormal{TE}_3(2,1)$ = 6}&~500 &-0.011 &0.182
  & &0.337 &0.176\\
  &1000 &-0.003 &0.133
  & &0.346 &0.128 \\
  &2000 &0.004 &0.093  & &0.351 &0.089\\
  \hline
\end{tabular}
\label{tab:sim2}
\end{threeparttable}}
\end{table}

\subsection{Application to real data}

In this section, we apply the proposed approach to a real customer loyalty data set from a telecom company in China.
The customer loyalty analysis tries to discover key factors affecting loyalty and the affecting process, with which
to choose proper actions to maintain customers.
In this study, 18553 randomly chosen customers answer questionnaire via personal interview or an online survey.
We drop the incompletely observed participants, leaving a total of 9833 participants.
In the questionnaire, customers scored their satisfaction about 13 different specific factors, respectively, such as network quality, tariff plan, voice quality, service quality, and so on. The scores are integers ranging from 1 to 10, with 1 denoting ``not satisfied at all'' and 10 denoting ``satisfied''. For the simplicity, we choose two important treatment factors, network quality and tariff plan, as an example to illustrate our approach, which are denoted by $Z_1$ and $Z_2$, respectively.
For each customer, a score indicating the loyalty to the company has also been obtained, and we use the score as the outcome variable $Y$.
Investigators also collected the information of each customer's  general satisfaction about the company, and it is used as the mediator variable $M$. We aim to assess whether the treatment factors have significant effects on improving customers' loyalty. In particular, we wish to figure out whether the effects of the factors on loyalty are mediated by the general satisfaction about the company.

We use the proposed approach for evaluation,
and  also consider the traditional approach without accounting for possible latent confounders for comparison.
For both approaches, we calculate estimators of  the average natural direct effect
 $NDE_j(2,1)$, the average natural indirect effect $NIE_j(2,1)$,
 and the average total effect $TE_j(2,1)$, and also their 95\% confidence intervals.  Of note, the average total causal effect $TE_j(2,1)$ is estimated by the sum of the estimated average natural direct effect $NDE_j(2,1)$ and indirect effect $NIE_j(2,1)$.
Subscript `$j$' indicates the corresponding causal effects of network quality and tariff plan for $j=1,2$.
We report these results in Table \ref{example:tab:results}.

\begin{table}[htbp]
\centering
\resizebox{0.48\textwidth}{!}{
\begin{threeparttable}
\def~{\hphantom{0}}
\caption{Results for the proposed approach and traditional approach on the analysis of the real customer loyalty data, respectively.
}
\centering
\begin{tabular}{ccccccc}
  \hline
  \addlinespace[1.5mm]
  \multirow{ 2}{*}{}&  \multicolumn{2}{c}{Proposed} &  &\multicolumn{2}{c}{Traditional } \\
  \addlinespace[1.5mm]
   \cline{2-3}
  \cline{5-6}
  \addlinespace[2mm]
   &  Estimate & 95\% CI  &   & Estimate & 95\% CI \\
      \addlinespace[2mm]
  \hline
   \addlinespace[1mm]
   Network quality & & & & &\\
   \addlinespace[1mm]
   \cline{1-1}
   \addlinespace[1.5mm]
   $\widehat{NDE_1(2,1)}$  &0.426  &(0.392, 0.458)  &  & 0.267 & (0.244, 0.291)  \\
   \addlinespace[2mm]
   $\widehat{NIE_1(2,1)}$  &0.039  &(0.028, 0.050)  & & 0.101 &(0.089, 0.113)  \\
      \addlinespace[2mm]
   $\widehat{TE_1(2,1)}$  &0.465 & (0.437, 0.492) &  & 0.369 &(0.345, 0.391)  \\
   \addlinespace[1.5mm]
   \hline
    \addlinespace[1mm]
   Tariff plan & & & & &\\
   \addlinespace[1mm]
   \cline{1-1}
   \addlinespace[1.5mm]
   $\widehat{NDE_2(2,1)}$  &0.401  &(0.373, 0.430)  &  & 0.271 &(0.250, 0.293)  \\
      \addlinespace[2mm]
   $\widehat{NIE_2(2,1)}$  &0.042 & (0.031, 0.053) &  & 0.110 &(0.099, 0.120)  \\
   \addlinespace[2mm]
   $\widehat{TE_2(2,1)}$  &0.443 & (0.419, 0.466) &  & 0.381 &(0.361, 0.403)  \\
   \addlinespace[1.5mm]
   \hline
 \end{tabular}
\label{example:tab:results}
\begin{tablenotes}
\small
\item CI: confidence interval.
\end{tablenotes}
\end{threeparttable}}
\end{table}

From Table~\ref{example:tab:results}, we can see that both the proposed approach and the traditional approach indicate positive and statistically significant causal effects of network quality and tariff plan on loyalty.
Since there are no latent confounders between the treatments (network quality and tariff plan) and outcome variable (loyalty)  by the randomization of this study, a simple linear regression of the outcome against the treatments can induce good estimates of the average total causal effects.
These results are exactly the same as the corresponding ones
obtained from the proposed approach,  but different from those obtained from the traditional approach.
From this point of view, the results for our proposed approach in Table~\ref{example:tab:results} are more reliable, compared with those obtained from the traditional approach which does not consider possible latent confounders between the general satisfaction and loyalty.
According to the results,
both the estimated average natural direct effects of network quality and tariff plan on loyalty are positive but larger than the corresponding indirect effects through the mediator variable (the general satisfaction about the company). This means that the causal effects of satisfaction about network quality and tariff plan on loyalty are mostly operated in a direct way, but only a small fraction of them are intermediated by the general satisfaction about the company.
Note also that the estimated total causal effect of network quality is a bit larger than that of tariff plan, which to some extent implies that the network quality may play a more important role in affecting loyalty than tariff plan.

\section{Conclusion}\label{sec:conclusion}

In this article, we have proposed an approach for causal mediation analysis with multiple treatments and latent confounders between the mediator and outcome variables. We give the formal definitions of the causal mediation effects in the multiple-treatment setting.
For the identification of these causal effects, we first postulate a causal model for potential outcomes and express the causal effects as functions of parameters of the causal model. Then we show the equality between parameters of the causal models with analogous parameters of the more feasible SEMs which is built upon observed variables. Using the idea of instrumental variable methods, we
provide sufficient conditions for identifying parameters of the SEMs,
which in turn gives the identification of the causal effects of interest.
Finally, we also develop an effective approach for estimation and inference.

\bibliographystyle{aaai}
\bibliography{mybib}

\end{document}